\documentclass[a4paper,12pt]{article}
\usepackage[cp1250]{inputenc}
\usepackage[T1]{fontenc}
\usepackage{amsfonts}
\usepackage{mitpress}
\usepackage{graphicx}
\usepackage{amsmath}
\usepackage{amsthm}
\usepackage{mathrsfs}
\usepackage{color}
\usepackage{lscape}
\usepackage{xfrac}
\usepackage[margin=3.5cm]{geometry}

\newcommand{\X}{\mathcal{X}}
\newcommand{\1}{\mathbf{1}}
\newcommand{\LL}{L_{\mu\nu}}
\newcommand{\LLL}{L_{\mu\nu}^I}
\newcommand{\CC}{C_{\mu\nu}}
\newcommand{\CCC}{C_{\nu\mu}}

\newcounter{licznik_t}

\newcounter{licznik_l}

\newcounter{licznik_w}
\newcounter{licznik_u}

\newcounter{licznik_pr}

\newtheorem{tw}[licznik_t]{Theorem}
\newtheorem{lm}[licznik_l]{Lemma}

\newtheorem{uwaga}[licznik_u]{Remark}
\newtheorem{wniosek}[licznik_w]{Corollary}

\theoremstyle{definition}

\newtheorem{przyklad}[licznik_pr]{Example}

\renewenvironment{proof}[1][Proof]{\mbox\\\noindent\textit{#1.} }{\hfill\qed\\}

\newdimen\dummy
\dummy=\oddsidemargin
\newcommand{\E}{\mathsf{E}}
\newcommand{\R}{\mathbb{R}}

\begin{document}

\author{Wioletta Szeligowska,
\quad Marek Kaluszka\thanks{
Corresponding author. E-mail adress: kaluszka@p.lodz.pl; tel.: +48 42
6313859; fax.: +48 42 6363114.}\\
\medskip \emph{\ {\small {Institute of Mathematics, Lodz University of
Technology, 90-924 Lodz, Poland}}\\
 }}
\date{\today}
\title{On Jensen's inequality for generalized Choquet integral with an application
to risk aversion
}
\maketitle


\begin{abstract} In the paper we give necessary and sufficient conditions for the Jensen inequality to hold for the generalized Choquet integral with respect to 
a pair of capacities. Next, we apply obtained result to the theory of risk aversion 
by providing the assumptions on  utility function and capacities under which 
an agent is risk averse. Moreover, we show that the Arrow-Pratt theorem can be generalized to cover the case, where the expectation is replaced by the generalized Choquet integral. 
\end{abstract}


\section{Introduction}
Let $(\Omega,\mathcal F)$ be a~measurable space, where $\mathcal F$ is a~$\sigma$-algebra of subsets of a~non-empty set $\Omega .$ 
A~{\it (normalized) capacity} on $\mathcal{F}$  is a set function $\mu\colon \mathcal{F}\to [0,1]$ such that $\mu(\emptyset)=0,$ $\mu(\Omega)=1$ and $\mu(A)\le\mu(B)$ whenever $A\subset B.$ Capacities are also called {\it fuzzy measures}, {\it nonadditive measures} or {\it monotone measures} $\cite{zwang1}.$ 
We write $\overline{\mu}$ for the {\it conjugate} or {\it dual} capacity of  $\mu ,$ that is, 
$\overline{\mu}(A)=1-\mu (A^c)$, where $A^c=\Omega\backslash A. $
The notion of capacity was introduced by Gustave Choquet in 1950 and has played an important role in fuzzy set theory, game theory, the rank-dependent expected utility model, the Dempster–Shafer theory,  and many others \cite{den,zwang1}.

The generalized Choquet integral is defined as 
\begin{align}\label{choquet}
\CC(X)=\int _0^\infty\mu (X>t)dt-\int _{-\infty}^0\nu(X<t)dt,
\end{align} 
provided that at least one of the two improper Riemann integrals is finite. 
Hereafter, $\mu (X\in A)=\mu (\{\omega\colon X(\omega)\in A\}).$
It is introduced by Tversky and Kahneman \cite{kah} for discrete random
variables and is used to describe the mathematical foundations of Cumulative
Prospect Theory.
Two outstanding examples
of  generalized Choquet's integral are: the {\it Choquet integral}  $C_{\mu}:=C_{\mu\overline{\mu}}$ and the {\it symmetric Choquet integral} $\breve{C}_\mu:=C_{\mu\mu},$ also known as the {\it $\breve{\hbox{S}}$ipo$\breve{\hbox{s}}$ integral} (see \cite{den,grabisch,zwang1}). If $\mu=\nu=P$ with a probability measure $P$, then the generalized Choquet integal reduces to the expectation  $\E _P X.$ 

The Jensen inequality says that  $\E f(X)\le f(\E X)$ for any concave function $f$ and for any random variable $X$ with a finite expectation $\E X.$ This is  one of fundamental result of the measure theory, having enormous applications in probability theory, statistics and other branch of mathematics. A lot of an extensions of this inequality is known with some additional assumptions on the function $f$ and for  different integrals  (see e.g. \cite{Girotto,ka,ka5, mit,Nic,pap,proschan,rom,st} and the references therein). To the best of our knowledge,  the Jensen inequality  
for the integral \eqref{choquet} has not been considered so far.  

The paper is organized as follows. In Section 2, we derive necessary and sufficient conditions for the capacities $\mu,\nu$  and function $f$ so that the Jensen inequality for the generalized Choquet integral holds.  In Section 3, two fundamental result of the risk theory  are
extended. The first one deals with existence of risk aversion 
and the second is the Arrow-Pratt theorem. 
The Appendix contains several commonly encountered examples of capacities which will be useful in Sections 2 and 3.

\section{Main result}\label{2}

Throughout the paper,  we denote by $I$ any open interval containing 0, bounded or not. Write $\LLL$ for the set of such measurable functions $X\colon \Omega \to I$ that $\CC(X)\in I.$ To prove our main result, we need the following lemma.
 
\begin{lm} Given any capacities $\mu,\nu$, the  integral $\eqref{choquet}$ has the following properties:
\begin{enumerate}
\item[(C1)] If   "$>$" or "$<$" in \eqref{choquet} is replaced by  "$\ge $" or  "$\le $", respectively,  then 
the value of generalized Choquet integral does not change.
\item[(C2)] $\CC(X)\le \CC(Y)$  
if $X(\omega)\le Y(\omega)$ 
for all $\omega$,  
\item[(C3)] $\CC(bX)=b\CCC(X)$ for $b\le 0$ and $\CC(bX)=b\CC(X)$ for all $b>0$ and $X\in \LL,$  
\item[(C4)] 
$\CC(a+X) =a+\CC (X)+
\int\limits_{-a}^{0}\bigl(\mu(X>s)-\overline{\nu}(X\ge s)\bigr) ds
$ for all $a\in \R.$ 
\end{enumerate}
\end{lm}  
      
\begin{proof}
(C1) Let  $\mu (X>t)<\mu(X\ge t)$ for some $t\in \R$. Then 
$$
\mu(X\ge t)>\mu (X>t)\ge \lim_{s\to t,s>t}\mu (X\ge s),
$$
so $t$  is a point of discontinuity of the function $f(s):=\mu(X\ge s)$. 
Since $f$ is increasing,  the set of points where $f$ is not continuous is at most countable, so $\mu (X>t)=\mu(X\ge t)$ almost everywhere.
Similar reasoning shows that $\nu (X<t)=\nu(X\le t)$ almost everywhere.

(C2) If $X\le Y$, then by monotonicity of $\mu$ and $\nu$ we have $\mu (X>t)\le \mu(Y>t)$ and $\nu (X<t)\ge \nu (Y<t)$ for each $t,$ which implies property  C2.  

(C3) It is clear that $\CC(0\cdot X)=0\cdot \CC(X)$. For  $b>0$ we have 
\begin{align*}
\CC(bX)=\int _0^\infty\mu (X>t/b)dt-\int _{-\infty}^0\nu(X<t/b)dt=b\CC(X),
\end{align*} 
By similar reasoning, $\CC(bX)=b\CCC(X)$ for $b<0$.

(C4) For an arbitrary $a\in \R$ we have
\begin{align*}
\CC(a+X)&=\int _{-a}^\infty\mu (X>s)ds-\int _{-\infty}^{-a}\nu(X<s)ds\\
&=\CC(X)+\int _{-a}^0\mu(X>s)ds-\int _0^{-a}\nu (X<s)ds\\
&=\CC(X)+a+\int_{-a}^{0}\bigl(\mu(X>s)-\overline{\nu}(X\ge s)\bigr) ds.
\end{align*} 
\end{proof}

First we establish  conditions under which the Jensen inequality holds. 

\begin{tw}\label{tw1}
Assume that   $f\colon I\to \R$ is an increasing and concave function and $f(0)\ge 0.$ Then  the following Jensen inequality holds for all $X\in \LLL$
\begin{align}\label{Jensen} 
\CC(f(X))\le f(\CC (X))
\end{align}
 if and only if 
 $\mu\le \overline{\nu}$, that is, $\mu(A)\le \overline{\nu}(A)$ for all $A.$
\end{tw}
\begin{proof} First we show that the  inequality $\eqref{Jensen}$ holds for all $X\in \LLL$  if and only if the following condition is valid:

(B) for all
$X\in \LLL$ and  $a\ge 0$  
\begin{align}\label{Jen1}
\CC(a+X)\le a+\CC(X).
\end{align}
Of course, condition $B$ is necessary for $\eqref{Jensen}$ to hold as  $f(x)=a+x$ is an increasing and concave function and $f(0)\ge 0$. 
We shall prove that the condition $B$ is also sufficient. 
Since  $f$ is increasing and concave, we have for  $x,y\in \R$ 
\begin{align} \label{nier1}
f(y) \leq f(x)+f'(x)(y-x),
\end{align}
where $f'(x)$ is the
right-sided derivative of $f$ at $x$.
As $f$ is concave on the open interval $I$, the derivative $f'(x)$ exists and is finite for all $x\in I$ (see \cite{kuczma}).   Moreover,  
$f(x)-xf'(x)\ge 0$; to see this,  put $y=0$ in \eqref{nier1}.  
From \eqref{nier1}, condition B and properties C2 and C3, we get
\begin{align}\label{jen2}
\CC(f(X))&\le \CC(f(x)-f'(x)x+f'(x)X) \nonumber \\
&\le f(x)-f'(x)x+\CC(f'(x)X)\nonumber\\
& = f(x)-f'(x)x+f'(x)\CC(X)
\end{align}
for all $x\in I$. Substituting  $x=\CC(X)$ in \eqref{jen2} we obtain  $\eqref{Jensen}$, so condition B is sufficient for $\eqref{Jensen}$  to hold.

It follows from condition B and properties C1 and C4 that  $\eqref{Jensen}$ is fulfilled for all  $X\in\LLL$ if and only if  for all $X\in \LLL$ and $a>0$ we have  
\begin{align}\label{w011}
\int_{-a}^{0}\left(\mu (X\geq
s)-\overline{\nu}(X\geq s)\right) ds\le 0.  
\end{align}
Put $X=b\1 _{A^c},$ where $-a<b<0$ and $A$ is any measurable set.  Then  
the inequality  $\eqref{w011}$ is of the form  $(-b)(\mu(A)-\overline{\nu}(A))\le 0$, and so  $\eqref{w011}$ holds for any $A$  if and only if $\mu (A)\le \overline{\nu}(A)$ for all $A$.
\end{proof}

For the Choquet integral $C_{\mu}(X)$ the condition $\mu\le \overline{\nu}$ is satisfied for all capacities $\mu,\nu$. Moreover, for the symmetric Choquet integral $\breve{C}_\mu(X)$ the condition holds, e.g. if $\mu$ is any superadditive measure (that is, $\mu (A)+\mu (B)\le \mu (A\cup B)$ for all $A,B$) or uncertainty measure  (see the Appendix, Example \ref{ex5}).  

Now we move to characterization the functions $f$ for which the inequality $\eqref{Jensen}$ holds. First, consider the case when both $\mu$ and $\nu$ are  $\{0,1\}$-valued capacities, that is, $\mu(A),\nu(A)\in \{0,1\}$ for all $A.$  

\begin{tw} \label{tw2}  Assume that $\mu,\nu$ are $\{0,1\}$-valued capacities and  $f\colon I\to\R$ is an increasing and continuous function such that  $f(0)=0.$ 
Assume also that $\mu (B)=1$ and $\nu (B^c)=1$ for some $B$. Then $\CC(f(X))\le f(\CC(X))$ for all $X\in \LLL$ if and only if $f$ is weakly superadditive, that is, $f(a)+f(b)\le f(a+b)$ for $a\le 0\le b$.
If  $\mu (B)=0$ or $\nu (B^c)=0$ for an arbitrary $B,$ then   
the Jensen inequality holds true without any extra assumptions on $f$.  
\end{tw}
\begin{proof}  
Put $b_X=\inf\{t\ge 0\colon\mu (X>t)=0\}$ and 
$a_X=\sup\{t\le 0\colon \nu(X<t)=0\}$ with  $X\in \LLL$. Clearly,   $a_X,b_X\in I$ and 
$\CC(f(X))=a_{f(X)}+b_{f(X)}$. By continuity and monotonicity of $f$, we can see that  
\begin{align*}
b_{f(X)}&=\inf \{f(s)\ge 0\colon\mu (f(X)>f(s))=0\}\\
&=\inf \{f(s)\ge 0\colon\mu(X>s)=0\}=f(b_X)
\end{align*} 
and $a_{f(X)}=f(a_X).$
Hence, the inequality $\CC(f(X))\le f(\CC(X))$ holds for all $X\in \LLL$
if and only if $f(a_X)+f(b_X)\le f(a_X+b_X)$ for all $X\in \LLL$. Observe that 
if  $\mu (B)=\nu (B^c)=1$ for some $B$, then $B\notin \{\emptyset,\Omega\}$ and 
we have $a_X=a$ i $b_X=b$ for 
$X=b\1 _B+a\1 _{B^c}$ with  any $a\le 0\le b$, and so $\eqref{Jensen}$  holds if and only if $f$ is weakly superadditive. 

Suppose it is not true that $\mu (B)=1$ and $\nu (B^c)=1$ for some $B$. Therefore, it is impossible that  $a_X<0$ and $b_X>0,$ and so $\eqref{Jensen}$ is valid for any $f$ as it has the form $f(a_X)\le f(a_X)$ or $f(b_X)\le f(b_X).$   
\end{proof}

As far as we know, the class of weak superadditive functions has not been 
examined in the literature so far. Recall that $f$ is {\it superadditive} if 
$f(x)+f(y)\le f(x+y)$ for all $x,y\in \R$.
Clearly, each superadditive function is weakly superadditive, but not vice versa. 
For instance, let  $h(x)$ be any increasing function for $x<0$ and let $h(x)=0$ for  $x\ge 0$. 
The function $f(x)=x+h(x)$ jest weakly superadditive, but it does not have to be  
superadditive.  
The function $f(x)=x-\sqrt{(-x)_+}$  is weakly superadditive, but is not concave
and  $f(x)=\ln (1+x)$ is not weakly superadditive and concave. Still, there exist weakly superadditive and concave functions, e.g.   
$f(x)=1-e^{-x}$ or $f(x)=x-(-x)_+.$  
Note that each of the functions given above is  increasing, continuous and vanishes at zero. 

\begin{uwaga}\emph{
An equality holds in $\eqref{Jensen}$ if  
\begin{itemize}
\item $\mu,\nu$ are arbitrary capacities and $X(\omega)=c$ for some $c\in I$ and for all $\omega$; \\
\item $f(x)=ax$ or $f(x)=b$ for all $x\in I$, where $a\ge 0$ and $b\in I$; 
\item $\mu=\overline {\nu}$ and $f(x)=ax+b$ for $a,b\in \R.$
\end{itemize}
}  
\end{uwaga}

Theorem  \ref{tw2} shows that even if $\mu=\overline{\nu},$ this not guarantee that from the Jensen inequality 
$\CC(f(X))\le f(\CC(X))$  it follows that  $f$ is a concave function.  
In fact, let us take  $\mu(A)=1$ for all $A\neq\emptyset$  
and  $\nu(A)=0$ for each $A\neq \Omega.$ It is obvious that  
$\mu=\overline{\nu }.$ Since it is not true that there exists $B$ such that  
$\mu (B)=1$ and $\nu (B^c)=1,$  Theorem $\ref{tw2}$ implies that the Jensen inequality also holds for nonconcave functions $f.$
Therefore, we have to add an extra assumption on capacities $\mu,\nu$.
 
\begin{tw}\label{tw3}
Suppose that $\mu\le \overline{\nu}$ and there exists a measurable set $B$ such that 
$\mu (B)>0$ and $\nu (B^c)>0.$ 
Given an increasing and continuous function $f\colon I\to\R$ such that $f(0)=0,$ if $\CC(f(X))\le f(\CC(X))$ for all $X\in \LLL$ that takes two values, then $f$ is concave.
\end{tw}
\begin{proof} Set  $p=\mu(B)$  and $q=\nu(B^c).$ Then
$0<p\le \overline{\nu }(B)=1-q<1,$ so 
$p,q\in (0,1).$  Put  $X=a+(b-a)\1_B,$ where  $\1_B$ denotes the indicator function of set $B.$  As  $f$ is increasing and $f(0)=0$, we have for $0\le a<b$  
\begin{align*}
\CC(f(X))=\int _0^{f(b)}\mu(f(X)>t)dt=f(a)(1-p)+f(b)p.
\end{align*} 
Clearly,  
$\CC(X)=a(1-p)+bp,$ so from the Jensen inequality we get 
$f(a)(1-p)+f(b)p\le f(a(1-p)+bp)$ for  $0\le a<b$ with fixed value of  $p\in (0,1)$. The function $-f$ is $p$-convex (see 
$\cite[\hbox{p.53}]{proschan}$). Hence, 
$-f$ is $J$-convex (see $\cite{kuhn}$ and also $\cite{daroczy}$
for an elementary proof). Since $f$ is continuous, $f$ is concave for $x\ge 0$ (cf. $\cite[\hbox{p}.133]{kuczma}$). 

For $a<b\le 0$, we have  
\begin{align*}
\CC(f(X))=-\int _{f(a)}^0\nu(f(X)<t)dt=f(a)q+f(b)(1-q).
\end{align*}
By the Jensen inequality, we have  $f(a)q+f(b)(1-q)\le f(aq+b(1-q)),$ so 
$f$ is concave for $x\le 0$.  

If $a<0<b,$ then  
$C _{\mu}(f(X))=f(a)q+f(b)p$, so the Jensen inequality is of the form  
\begin{align}\label{for0}
f(a)q+f(b)p\le f(aq+bp).
\end{align}
As $f$ is continuous and concave for  $x\le 0$  and  for $x\ge 0$, there exist the left-hand side derivative $f'_-(0)$ and the right-hand side derivative $f'_+(0).$  
Substituting $b=-aq/p$ in \eqref{for0}, dividing both sides of $\eqref{for0}$ by $a<0$ and taking the limit as $a$ goes to zero, we obtain 
the inequality  $f'_-(0)\ge f'_+(0),$ so the function $f$ is concave on the interval $I.$
\end{proof}

\begin{uwaga}\emph{
Note that the assumptions of Theorem $\ref{tw3}$ imply that neither $\mu$ nor $\nu$ is a $\{0,1\}$-valued capacitity (see the proof of Theorem $\ref{tw3}$). As we showed above, we cannot omit the assumptions about the existence of such set $B$  
that $\mu(B)>0$ and $\nu(B^c)>0$.
}
\end{uwaga}

We now give a counterpart of Theorems \ref{tw1} and \ref{tw3} under the assumption that the inequality $\eqref{Jensen}$ holds only for nonnegative functions $X.$  

\begin{tw}\label{tw4}
Suppose that  
$f\colon I\to\R$ is 
an increasing and continuous function  and $f(0)=0.$ 
Assume there exists $B$ such that  
$0<\mu (B)<1$. Then 
the inequality $\eqref{Jensen}$  holds for all $X\in L_{\mu\nu}^{[0,\infty)}$ if and only if $f$ is concave for $x\ge 0.$  
Moreover, if $\mu $ is a $\{0,1\}$-valued capacity, then $\eqref{Jensen}$ is satisfied for any $X\in L_{\mu\nu}^{[0,\infty)}$ without any additional assumption on $f.$
\end{tw}
\begin{proof} The proof follows almost exactly as for Theorems \ref{tw1} and \ref{tw3}, so we omit it.  
\end{proof}

In Theorems  \ref{tw1}-\ref{tw4} we restrict our attention  to the case of an increasing and continuous function $f.$ An extension to a wider class of function is a difficult task  due to the property  $C3$ of the generalized Choquet integral. To the best of our knowledge, the only result in this direction  is that of  Girotto and Holtzer $\cite[\hbox{Theorem 2.5}]{Girotto},$ where a  Jensen type inequality for the integral $C_\mu$  was obtained.

\section{Application}
\label{sek2}

Consider an agent which has a nonnegative reference point $w$, e.g. an initial wealth measured in monetary terms.
The agent faces a random outcome $X$ with known distribution. In this
background, $X$ may be a random variable which takes both negative and positive values, which means
that there is a possibility of yielding a gain from investition.   
If the agent  decides to buy a contract for some premium $\pi$ 
and the outcome  $X$ occurs, then it will be reimbursed  and it does not influence the wealth. On the other hand,
if she decides not to buy the contract, her wealth is $w-X$. 

According to the von Neumann-Morgenstern theory, preferences of the agent can be described by a
continuous and strictly increasing utility function $u\colon I\to \R$ with $u(0)=0$. The most common
examples are: the {\it exponential utility function} $u(x)=1-e^{-ax}$, $a>0$;
the {\it power utility function} $u(x)=(x+a)^b-a^b$, $a\ge 0$, $b>0$;
the {\it logarithmic utility function} $u(x)=\ln ((x+a)/a)$, $a\geq 1$;
the {\it power-expo utility function} $u(x)=1-\exp(-bx^c)$, 
$b,c>0$ (see \cite{munk}).

Pratt $\cite{pratt}$ suggested to use the equivalent utility principle to
find the maximum premium $\pi$ which the agent is willing to pay. Then $\pi$
is the solution of the following equation 
\begin{align}\label{skladka}
u(w- \pi)=\mathbb{E} u(w-X).
\end{align}
We say that the agent is risk averse, if she is willing to pay more than $%
\mathbb{E} X$ for an insurance contract paying out the monetary equivalent
of a random outcome $X$, regardless her initial wealth $w$. One may prove
that the  agent is risk averse if and only if the function $u$ is concave,
that is, $u(ax+(1-a)y)\ge au(x)+(1-a)u(y)$ for all $x,y$ and $0<a<1$.
Moreover, if two agent have the same initial wealth and the $i$-th agent
has twice differentiable utility function $u_i,$ $i=1,2$, then the first of
them is more risk averse (wants to pay not less than the other agent), if
and only if $r_{u_1}(x)\ge r_{u_2}(x)$ for all $x$, where $r_{u_i}(x)$ is the \textit{coefficient of absolute
risk aversion} of the $i$-th agent defined as 
\begin{align}  \label{equation}
r_u(x)=- \frac{u^{\prime \prime }(x)}{u^{\prime }(x)}
\end{align}
for $u\in\{u_1,u_2\}.$ It is also called the \textit{Arrow-Pratt index}.
Similar results were obtained independently by Arrow $\cite{arrow}$ and de
Finetti $\cite{monte}$. It is worth to note that in many books
and papers, their authors give only a heuristic explanation of this result
based on the approximation of premium $\pi$ by using the Taylor formula, see
e.g. $\cite{ek,pena}$. Precise proofs can be found in $\cite%
{follmer,munk,pratt}.$ 
The studies on risk aversion were pursued by many researchers, see $\cite%
{follmer, golier, Machina,munk}$ among others.

The aim of the paper is to generalize the two aforementioned results by
replacing of the expected value with the generalized Choquet integral. 
Let $\pi_u(X,w)$ denote the premium determined from the following formula 
\begin{align}  \label{uzytecznosc}
u\bigl(w- \pi _u(X,w)\bigr)=\CC(u(w-X)),
\end{align}
where $X$ is a  financial outcome $X$ and $\CC(X)$  
is the Choquet integral with respect to capacities $\mu,\nu$.
Since the function $u$
is strictly increasing and continuous, 
the premium $\pi _u(X,w)$ exists and is determined uniquely if $w-X\in \X_u$, where $\X_u$ denotes the set of such measurable functions $X\colon \Omega \to \R$ that $\CC (X)\in I$ and $\CC(u(X))\in u(I)$.
The premium $\pi _u$ was proposed in $\cite{ka1,ka4}$ in case of
the capacities being Kahneman-Tverski distorted measures (see the Appendix, Example \ref{ex2}). A lot of properties of that premium was studied, but the problem of risk aversion measure was not examined.  

Now, we are ready to extend  the Arrow-Pratt result.
We say that an
agent is \textit{risk neutral}, if the utility of her money is measured by
the means of its value, that is, if her utility function is $u_0(x)=x$ for
all $x\in \mathbb{R}$. We say that an agent with utility function $u\colon I\to \R$ is \textit{risk averse}, if for all $w \ge 0$ and $X$  such that $w-X \in 
\X_u$ we have 
\begin{align}  \label{awersja0}
\pi _u(X,w) \ge \pi _0:=\CCC(X)+\int _0^w\bigl(\overline{\nu}(X<s)-\mu(X<s)\bigr)ds,
\end{align} 
where $\pi _0$ is the premium of a risk neutral agent.

\begin{tw}
\label{tw5} An agent with
a concave utility function is  
risk averse if and only if $\mu\le \overline{\nu}.$ 
Moreover, if an agent is risk averse, $\mu\le \overline{\nu}$ and 
$\mu (B),\nu(B^c)>0$ for some $B$, then 
$u$ is a concave function. 
\end{tw}

\begin{proof} Of course, $\pi _u(X,w)=w-u^{-1}(\CC(u(w-X)))$. An agent with the function $u$ is risk averse if and only if  for all $w\ge 0$ such that  $w-X\in \X_u$ we have
\begin{align}\label{kk}
\CC(u(w-X)) \leq u(\CC(w-X)).
\end{align}  
The desired result is now a simple consequence of Theorems \ref{tw1} and \ref{tw3}. 
\end{proof}

\begin{wniosek} Assume that an agent with an utility function $u$ uses the Choquet integral $C_\mu$
to evaluation and $\mu$ is not a $\{0,1\}$-valued capacity. Then 
the agent is  
risk averse if and only if $u$ is a concave function. 
\end{wniosek} 

Given $u,v\colon I\to \R$ such that $v(I)\subset I$, we say that an agent with the utility function $u$ is \textit{more risk averse} than
an agent with utility function $v$, if for all $w\ge 0$ such that $w-X\in \X_u\cap \X _v,$ we have 
\begin{align}
\pi _u(X,w) \geq \pi _v(X,w).
\end{align}
Theorem \ref{tw6} provides a
generalization of the Arrow-Pratt Theorem in the case, when $\mu$ and $\nu$  are capacities. 
Recall that $r_u(x)$  denotes the coefficient of
the absolute risk aversion $\eqref{equation}$.

\begin{tw}
\label{tw6} Assume that $\mu\le \overline{\nu}$ and $\mu(B),\nu (B^c)>0$ for some $B$. 
Let $r_u$ and $r_v$ be the coefficients of the absolute risk
aversion of concave  and twice differentiable utility functions $u,v$, respectively. The following conditions are
equivalent:
\begin{description}
\item $(i)$ an agent with the utility function $u$ is more risk averse than
an agent with utility function $v$;
\item $(ii)$ $u=g \circ v$ for some strictly increasing and concave function $g$;
\item $(iii)$ $r_u(x) \geq r_v(x)$ for all $x\in I$.
\end{description}
\end{tw}
\begin{proof}
$\left( i\right) \Rightarrow \left( ii\right)$:
From  the assumption $\pi_u(X,w) \geq \pi_v(X,w)$ it follows that 
for all $w-X\in \X_u\cap \X_v$ we have
\begin{align}\label{pon}
\CC(g (Y)) \leq g(\CC (Y)),
\end{align}
where $g=u \circ v^{-1}$ and $Y=v(w-X).$ 
From Theorem \ref{tw3} we conclude that $g$ is concave.
\\\medskip 
$\left( ii\right) \Rightarrow \left( i\right)$:
Let $u=g\circ v$. Then,
from the concavity of $g$ and from Theorem \ref{tw1}, for all $Y\in \LL ^{v(I)}$ we get
\begin{align}
u(w-\pi _u(X,w))&=\CC(g(Y)) \nonumber \\ &\le g(\CC(Y)) =u(w-\pi _v(X,w)),
\end{align}
thus $\pi _u(X,w)\ge \pi _v(X,w).$
\\ $\left( ii\right) \Leftrightarrow \left( iii\right)$:
Let $g=u \circ v^{-1}$. Function $g$ is increasing and twice differentiable, because it is a composition of functions $u$ and $v^{-1}$. Hence $g$ is concave if and only if 
$g^{\prime\prime} (x)\leq 0$ for $x\in v(I)$. Since for all $x$ we have
\begin{align}
g^{\prime\prime}(x)=-\frac{u^\prime(v^{-1}(x))}{[v^\prime(v^{-1}(x))]^2}(r_u(v^{-1}(x))-r_v(v^{-1}(x))) \leq 0,
\end{align}
so $g$ is concave.
\end{proof}

Note that our proof is different from those in $\cite{follmer,munk,pratt},$
where  some extra unnecessary
assumptions are added, e.g. about the continuity of the Arrow-Pratt coefficient or the strict convexity of the utility function. 

It was Georgescu and Kinnunen $\cite{geor-kinu}$ and Zhou et al. $\cite{zhou}$ who first analyze the risk aversion
within the framework of the Liu uncertainty theory (see $\cite{bliu11})$.
They introduce the
notions of uncertain expected utility, uncertain risk premium and give a
counterpart of the Arrow-Pratt theorem under the uncertainty theory.
The heuristic justification of the Arrow-Pratt Theorem was based on the
Taylor approximation of premium.
Theorems $\ref{tw5}$ and $\ref{tw6}$ generalize the results of $\cite{geor-kinu,zhou}$ for arbitrary capacities with proofs which do not use the heuristic reasonings (see Remark $\ref{aproksymacja}$).

\begin{przyklad}
Suppose that  $\mu(A)=g(P(A))$ and $\nu(A)=h(P(A))$ for each $A$, where $g,h$ are weighting  functions (see Appendix, Example \ref{ex2}). 
From many empirical research it follows that  $g$ is an $S-$shaped function, i.e. $g$ is concave on $[0,p]$ and is convex on  $[p,1]$ for some $p\in (0,1)$. Moreover, small probabilities are overestimated and high probabilities are underestimated. The function $h$ has the same form as $g$, but it has a bit different parameters. 

Tversky and Kahneman \cite{kah} propose the following weighting function
\[
g(p)=\frac{p^\gamma}{(p^\gamma+(1-p)^\gamma )^{1/\gamma}},
\]
where $\gamma\in (0.28,1)$. In the empirical research the following estimation of the parameter
$\gamma$ was obtained : $\gamma =0.61$ for  $g$,
$\gamma =0.69$ for  $h$ (see $\cite{kah}$). 
Figure 1 shows a graph of $g$ (solid line) and $\overline{h}$ (dashed line) with parametres suggested by Tversky and  Kahneman.
Since $\overline{h}\ge g$, we have  $\mu\le \overline{\nu}$. 
\begin{figure}[thb]
\begin{center}
\includegraphics[width=6 cm]{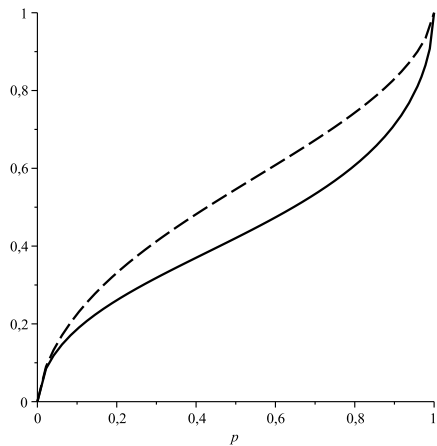}
\end{center}
\caption{The Kahneman-Tversky weighting functions}
\end{figure}

Goldstein and Einhorn \cite{Gold} introduce the function
\[
g(p)=\frac{\delta p^\gamma}{\delta p^\gamma+(1-p)^\gamma }.
\]
Abdellaoui \cite{Ab} estimated for it values in the empirical way: 
$\delta = 0.65$, 
$\gamma = 0.60$ for $g$ and  $\delta = 0.84$, $\gamma = 0.65$
for $h$. The function $g$  is depicted by a solid line and $\overline{h}$ is depicted by a dashed line in Figure 2. 
The condition $\mu\le \overline{\nu}$ is satisfied.

\begin{figure}[thb]
\begin{center}
\includegraphics[width=6 cm]{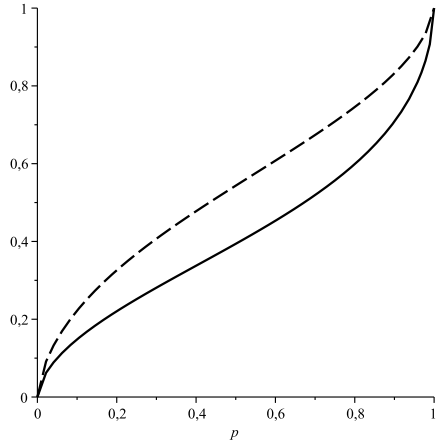}
\end{center}
\caption{The Goldstein-Einhorn weighting functions}
\end{figure}

Prelec \cite{prelec1} axiomatically derived the 
probability weighting function 
\begin{align*}
g(x)=\exp(-\delta(-\ln p)^\gamma),\quad \delta >0,\; 0<\gamma\le 1.
\end{align*}
Following Rieger and Wang \cite{prelec2}, we  use Perlec's
weighting function with the same parameters $\delta=1$ and $\gamma=0.74$ both for gains and for losses. The functions $g(x)$ and $\overline{h}(x)=1-g(1-x)$ are illustrated in Figure 3. The inequality $\mu\le \overline{\nu}$ is also valid. 

\begin{figure}[thb]
\begin{center}
\includegraphics[width=6 cm]{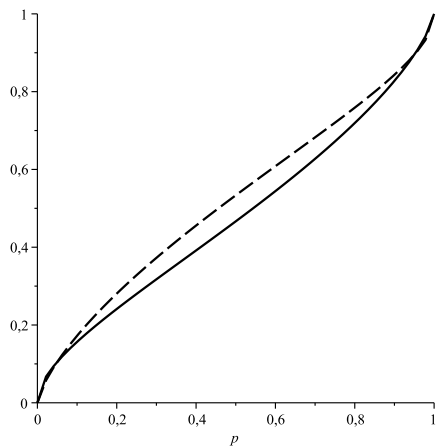}
\end{center}
\caption{The Prelec weighting functions}
\end{figure}
\end{przyklad}

\begin{uwaga}\label{aproksymacja}
\emph{The approximation is used in the classical utility theory with probability measure $\mu$  to justify the meaning of the coefficient of risk aversion 
$r_u.$ We provide similar approximation for premium $\pi _u(X,w).$ From the Taylor formula at point $w$ and from $\eqref{uzytecznosc}$
we obtain 
\begin{align*}
u(w)-u^{\prime }(w)\widehat{\pi}_u(X,w)&=\CC\bigl(u(w)+u^{\prime }(w)(-X)+%
\tfrac{1}{2}u^{\prime \prime }(w)X^2\bigr),\\
&=u(w)-u'(w)\CCC(Y)\\
&\hspace{.4cm}+\int _{-u(w)}^0\bigl(\mu (-u'(w)Y>s)-\overline{\nu}(-u'(w)Y>s)\bigr)ds,
\end{align*}
where $\widehat{\pi}_u(X,w)$ denotes an approximation of premium $\pi
_u(X,w) $ and $Y=X+r_u(w)X^2/2$.  After an elementary calculation we get the following 
\textit{de Finetti-Arrow-Pratt approximation}: 
\begin{align}  \label{appro}
\widehat{\pi}_u(X,w)=\CCC(Y)+\int _0^{u(w)/u'(w)}\bigl(\overline{\nu}(Y<t)-\mu(Y<t)\bigr)dt.
\end{align}
If an agent uses the Choquet integral  (that is, $\mu=\overline{\nu}$), then 
\begin{align}  \label{appro1}
\widehat{\pi}_u(X,w)=C_{\overline{\mu}}(Y)\ge C_{\overline{\mu}}(X)=\pi _0,
\end{align}
and the greater $r_u(w)$ measure, the greater the difference between the premium $\pi _u(X,w)$ and the premium of risk neutral agent is, because $Y=X+r_u(w)X^2/2$ and Choquet integral has the property C2. When the agent uses the symmetric Choquet integral $(\mu=\nu)$ and $\mu (A)+\mu (A^c)=1$ for all $A$ (e.g. for Liu uncertainty measure,  see Example \ref{ex5}), we have the same interpretation, because
\begin{align*} 
\widehat{\pi}_u(X,w)=\breve{C}_\mu(Y)\ge \breve{C}_\mu(X)=\pi _0.
\end{align*}
In the other cases, an interpretation of risk aversion measure via approximation \eqref{appro}  is not clear.
}
\end{uwaga}

Finally, we consider an  important special case, namely, when
the loss $X$ does not exceed agent's wealth. 

\begin{tw} Assume $\mu$ is not a $\{0,1\}$-valued capacity and $X(\omega)\le w$ for all $\omega\in \Omega$.  An agent with utility function $u$ is risk averse if and only if $u$ is concave for $x\ge 0$. Moreover, given 
concave  and twice differentiable utility functions $u$ and $v$,  an agent with the utility function $u$ is more risk averse than
an agent with utility function $v$ if and only if 
$r_u(x) \geq r_v(x)$ for all $x\in I\cap [0,\infty)$.
\end{tw}

\begin{proof} The proof is 
straightforward. This follows from Theorem \ref{tw4}. We omit the details. 
\end{proof}

\section{Appendix}

We give below a few commonly encountered examples of capacities. 

\begin{przyklad}\label{ex2} Let $g\colon [0,1]\to [0,1]$ be a such non-decreasing function that $g(0)=0$ and~$g(1)=1, $ called the {\it weighting function} or the {\it distortion}. The distortions of probability measure $P$ of the form $\nu(A)=g\big(P(A)\big)$ are capacities. They are crucial for~the Cumulative Prospect Theory proposed by Kahneman and Tversky $\cite{kah},$ which concerns the behavior of agents on financial market using psychological aspects. It is worthy to emphasize that Kahneman was awarded with Nobel Prize in Economic Sciences in 2002 for that theory. The capacities also became a basic tool to measure risk in insurance mathematics $\cite{follmer, ka1,ka4}.$ 
\end{przyklad}

\begin{przyklad}\label{ex1}
Let  $\mathcal{P}$ be a family of probability ditributions on $\Omega$ and let parameter $\theta\in [0,1]$ balances optimistic and pessimistic
attitude of an agent. The {\it Hurwicz capacity} of the form 
\begin{align*}
\mu^*(A)=\theta \inf_{P\in\mathcal{P}} P(A)+(1-\theta)\sup_{P\in\mathcal{P}}P(A),\end{align*} 
is used in theory of decision making, if we have only partial information on the distribution $P$ of random outcome (see $\cite{gul}$). 
\end{przyklad}

\begin{przyklad} The {\it possibility measure} is defined by 
$\mu (A)=\sup _{x\in A}\psi (x)$ for $A\neq \emptyset$, where $\psi\colon \Omega\to [0,1]$ is any such function that $\sup_{x\in \Omega}\psi(x)=1.$  The {\it necessity measure} is the conjugate of a possibility measure 
 $\cite{dub,zwang1}.$
\end{przyklad}

\begin{przyklad} Given a nonempty subset $K$ of $\Omega,$  we call $\mu_K$ an {\it unanimity game} if $\mu _K(A)=1$ for $K\subset A$ and $\mu _K(A)=0$ otherwise. The set $K$ is called a {\it coalition}. Each $\{0,1\}$-valued capacity on a finite set $\Omega$ can  be written as
a maximum of unanimity games  or as 
a linear combination of unanimity games with integer coefficients (see $\cite{den2}$).
\end{przyklad}

\begin{przyklad}\label{ex3}
A map $m\colon 2^\Omega \to [0,1]$ is called the {\it mass function}, if it satisfies conditions $m(\emptyset)=0$ and $\sum_{A\in 2^\Omega}m(A)=1$. 
The {\it belief function} and {\it plausibility function} are defined, respectively, as follows
\begin{align*}
\text{Bel}(A)=\sum_{B\subset A} m(B),\quad\quad \text{Pl}(A)=\sum_{B\cap A\neq\emptyset} m(B).
\end{align*}
Both functions are examples of capacities $\cite{zwang1}$. The maps defined above play a crucial role in the Dempster-Schafer~theory $\cite{19}.$
\end{przyklad}

\begin{przyklad}\label{ex5} A map  $M\colon\mathcal{F}\to [0,1]$ is called 
the {\it uncertainty measure}, if  
\begin{enumerate}
\item[(M1)] $M(\Omega)=1,$
\item[(M2)] $M(A^c)+M(A)=1$ for all  $A\in \mathcal{F},$ 
\item[(M3)]   
$M\Big(\bigcup_{n=1}^\infty A_n\Big)\le \sum_{n=1}^\infty M(A_n)$
for any sequence  $(A_n)\subset \mathcal{F},$
\end{enumerate}
This term was introduced by  Liu $\cite{bliu11},$
We will show that the uncertainty measure is a capacity. 
A consequence of~M1 and~M2 is that $M(\emptyset)=0.$ If $\Omega=A^c\cup B,$ \label{ozndopelnienie} where $A\subset B$,
then from ~M1, M3 and M2 it follows that 
\begin{align*}
1=M(X)=M(A^c\cup B)\le M(A^c)+M(B)=1-M(A)+M(B),
\end{align*} 
so $M$ is a monotone set function, as desired.  

A particular case of  uncertainty measure is  the {\it credibility measure}  Cr$(A)$ defined by
\begin{align}\label{ge}
\text{Cr}(A)=\Bigl(\sup_{x\in A} v(x)+1-\sup_{x\in A^c} v(x)\Bigr)/2, 
\end{align}
where $v\colon \Omega \to [0,1]$ is any function. 
\end{przyklad}

%

\end{document}